\documentclass[conference,letterpaper]{IEEEtran}

\usepackage[utf8]{inputenc} 
\usepackage[T1]{fontenc}
\usepackage{url}
\usepackage{ifthen}
\usepackage{cite}
\usepackage[cmex10]{amsmath} 

\usepackage{amsthm}
\usepackage{amssymb,amsfonts}
\usepackage{algorithmic}
\usepackage{graphicx}

\theoremstyle{plain}
\newtheorem{thm}{Theorem}

\newtheorem{cor}[thm]{Corollary}

\theoremstyle{definition}

\newtheorem{exmp}{Example}

\theoremstyle{remark}

\newcommand{\ex}{\mathbb{E}}

\newcommand{\df}{\mathrm{d}}

\newcommand{\PP}{\mathcal{P}}

\newcommand{\XX}{\mathcal{X}}

\newcommand{\YY}{\mathcal{Y}}

\newcommand{\FF}{\mathcal{F}}

\newcommand{\MM}{\mathcal{M}}

\newcommand{\PX}{\mathcal{P}(\mathcal{X})}

\begin{document}
\title{Properties of the Strong Data Processing Constant for R\'enyi Divergence} 


\author{%
  \IEEEauthorblockN{Lifu~Jin}
  \IEEEauthorblockA{EPFL\\
                    Lausanne, Switzerland\\
                    lifu.jin@epfl.ch} \and
   \IEEEauthorblockN{Amedeo Roberto Esposito}\IEEEauthorblockA{
Institute of Science and Technology Austria (ISTA)\\ Klosterneuburg, Austria\\
amedeoroberto.esposito@ist.ac.at}
\and
\IEEEauthorblockN{Michael Gastpar}\IEEEauthorblockA{ EPFL\\ Lausanne, Switzerland \\ michael.gastpar@epfl.ch}
}


\maketitle


\begin{abstract}
Strong data processing inequalities (SDPI) are an important object of study in Information Theory and have been well studied for $f$-divergences. Universal upper and lower bounds have been provided along with several applications, connecting them to impossibility (converse) results, concentration of measure, hypercontractivity, and so on. In this paper, we study R\'enyi divergence and the corresponding SDPI constant whose behavior seems to deviate from that of ordinary $\Phi$-divergences. In particular, one can find examples showing that the universal upper bound relating its SDPI constant to the one of Total Variation does not hold in general. In this work, we prove, however, that the universal lower bound involving the SDPI constant of the Chi-square divergence does indeed hold. Furthermore, we also provide a characterization of the distribution that achieves the supremum when $\alpha$ is equal to $2$ and consequently compute the SDPI constant for R\'enyi divergence of the general binary channel.
\end{abstract}

\section{Introduction}

The well-known data processing inequality (DPI) for relative entropy states that for any two probability distributions $\mu, \nu$ such that $\nu \ll \mu$ over an alphabet $\mathcal{X}$ and for any stochastic transformation $K$ with input alphabet $\mathcal{X}$ and output alphabet $\mathcal{Y}$, we have
\begin{equation}
    D_{\text{KL}}(\nu K|| \mu K) \leq D_{\text{KL}}(\nu || \mu),
\end{equation}
where $D_{\text{KL}}(\nu || \mu)$ is the KL-divergence between $\nu$ and $\mu$. This inequality can be improved if we fix $\mu, K$ and only let $\nu$ vary. Indeed, for some stochastic transformation $K$, one can have that for every $\nu$, unless $\nu = \mu$, $D_{\text{KL}}(\nu K|| \mu K)$ is strictly less than $D_{\text{KL}}(\nu || \mu)$. Formally, one can define the Strong Data Processing Inequality (SDPI) constant of the KL-divergence for $\mu, K$ as in \cite{ahlswede1976spreading, raginsky2016strong, polyanskiy2017strong}
\begin{equation}
    \eta_{\text{KL}}(\mu, K) := \sup_{\nu \neq \mu} \dfrac{D_{\text{KL}}(\nu K|| \mu K)}{D_{\text{KL}}(\nu || \mu)}.
\end{equation}
Moreover, $\Phi$-divergences, a well-known generalization of the KL-divergence~\cite{csizsar1967,liese06}, also allow for the definition of a corresponding SDPI constant. Considering a convex function $\Phi: \mathbb{R}^+ \to \mathbb{R}$ such that $\Phi(1)=0$ and two probability distributions $\nu \ll \mu$, one can define the $\Phi$-divergence as
\begin{equation}
    D_\Phi(\nu\|\mu) = \mathbb{E}_\mu\left[\Phi\left(\frac{d\nu}{d\mu}\right)\right].
\end{equation}
Given a Markov kernel $K$, the corresponding SDPI constant is defined as
\begin{align}
    \eta_{\Phi}(\mu, K) & := \sup_{\nu: \nu \neq \mu} \dfrac{D_{\Phi}(\nu K \| \mu K)}{D_{\Phi}(\nu \| \mu)}, \\
    \eta_{\Phi}(K) &:= \sup_{\mu \in \PP(\XX)} \eta_{\Phi}(\mu, K),
\end{align}
where $\PP(\XX)$ is the set of all the probability distributions over $\XX$. These objects have been connected to several others, such as the maximal correlation and the so-called hypercontractivity constants of certain Markov operators~\cite{Witsenhausen1975ONSO,anantharam2013maximal}. Moreover, given the importance of the Data Processing Inequality in Information Theory and the possibility of improving the corresponding results computing the SDPI constants, they have gained increasing interest over the years, leading to a variety of applications, universal upper and lower bounds~\cite{cohen1993relative,delMoral03,raginsky2016strong,polyanskiy2017strong}. 
In particular, it is known that for any $\Phi$ and any channel $K$~\cite[Theorem 3.1]{raginsky2016strong} 
\begin{equation}
    \eta_\Phi(K) \leq \eta_{\text{TV}}(K),  \label{eq:boundTV}
\end{equation} while if $\Phi$ is also three times differentiable, given any measure $\mu$~\cite[Theorem 3.3]{raginsky2016strong}:
\begin{align}
    \eta_\Phi(\mu,K) \geq  \eta_{\chi^2}(\mu,K) \text{ and, }
     \eta_\Phi(K) \geq  \eta_{\chi^2}(K). \label{eq:boundChi}
\end{align}
$\eta_{\text{TV}}$ and $\eta_{\chi^2}$ denote the SDPI constant of the divergences induced, respectively, by $\Phi(x)=\frac12|x-1|$ and $\Phi(x)=x^2-1$.

R\'enyi divergences represent another family of divergences that are known to satisfy the DPI and lend themselves to the definition of a corresponding SDPI constant. The R\'enyi divergence $D_{\alpha}(\nu||\mu)$ of two probability distributions $\nu \ll \mu$ for $\alpha \not \in \{0, 1, \infty\}$ can be defined as follows \cite{van2014renyi}
\begin{equation}
     D_{\alpha}(\nu \| \mu) := \dfrac{1}{\alpha - 1} \log \ex_{\mu}\left[\left(\dfrac{\df \nu}{\df \mu}\right)^{\alpha}\right].
\end{equation}
Moreover, one can define the corresponding SDPI constant
\begin{align}
    \eta_{\alpha}(\mu, K) & := \sup_{\nu: \nu \neq \mu} \dfrac{D_{\alpha}(\nu K \| \mu K)}{D_{\alpha}(\nu \| \mu)}, \label{eq:sdpiRenyi} \\
    \eta_{\alpha}(K) &:= \sup_{\mu \in \PP(\XX)} \eta_{\alpha}(\mu, K).
\end{align}
Notice that the R\'enyi Divergences do not belong to the family of $\Phi$-divergences. However, they can be connected to the Hellinger divergences via a $1$-$1$ mapping, see \cite[Eq. (80)]{sasonV16}. 
Despite this, its behavior appears to be different from that of $\Phi$-divergences. As a case in point,
it cannot be upper bounded as in Eq.~\eqref{eq:boundTV} by the SDPI constant of the Total Variation distance. An explicit example showing that we may have $\eta_{\alpha} > \eta_{\text{TV}}$ appears in~\cite[Example 2]{esposito2023concentration}. In general, one cannot leverage the results derived for the SDPI of $\Phi$-Divergences to provide bounds on the SDPI of $D_\alpha$. In particular, existing results do not imply a universal lower bound for $D_\alpha$, like the one in Eq. \eqref{eq:boundChi}.
Consequently, the SDPI constant of R\'enyi divergence stands as an interesting object of study. Indeed, it can be proven to be related to the log-Sobolev inequality and hypercontractivity of certain types of operators~\cite{gross75}. Characterizing it can lead to improved concentration results for non-independent random variables~\cite{esposito2023concentration} and improved lower bounds on the Bayesian risk in estimation procedures with privatized samples~\cite[Section 3.1]{esposito2023lower}. Moreover, a recent line of work is studying the contraction of R\'enyi Divergences along diffusion processes in order to provide convergence guarantees in the context of sampling~\cite{vempala2019Langevin,chen2022Sampling}. To conclude, the limiting behaviour of the R\'enyi divergence could provide new characterizations of $\eta_{\text{KL}}$, since we have by \cite[Theorem 5]{van2014renyi} that $
    D_{\text{KL}}(\nu||\mu) = \lim\limits_{\alpha \downarrow 1} D_{\alpha}(\nu||\mu). $

In this paper, we study said object. Notably, we can prove that the universal lower-bound does indeed hold and thus
\begin{equation}
    \eta_\alpha(K) \geq \eta_{\chi^2}(K).
\end{equation}
Moreover, we study some properties of $\eta_{\alpha}$ when $\alpha=2$ and we manage to characterize the distribution achieving the supremum in~\eqref{eq:sdpiRenyi}. We will explain why our theorems are promising with some simple but conceptually crucial examples.

\section{Preliminaries and Notations}

We denote by $\mathcal{P}(\mathcal{X})$ the set of all probability distributions on an alphabet (or a space) $\mathcal{X}$. The set of all real-valued functions on
$\XX$ is denoted by $\FF(\XX)$. Markov kernels (channels) $\{K(y|x): x \in \mathcal{X}, y \in \mathcal{Y}\}$ acts on probability distributions $\mu \in \mathcal{P}(\mathcal{X})$ from the right as follows
\begin{equation}
    \mu K(y) = \sum_{x \in \XX} \mu(x) K(y|x), \qquad y \in \YY,
\end{equation}
or on functions $f \in \FF(\XX)$ from the left as follows
\begin{equation}
    K f(x) = \sum_{y \in \YY} K(y|x) f(y), \qquad y \in \YY.
\end{equation}
The set of all such kernels is denoted as $\mathcal{M}(\YY|\XX)$.

We say a pair $(\mu, K) \in \PX \times \MM(\YY|\XX)$ is admissible if $\mu \in \PP_*(\XX)$ and $\mu K \in \PP_*(\YY)$, where $\PP_*$ denotes all strictly positive distributions. For any such pair, there exists a unique channel $K^* \in \MM(\XX|\YY)$ with the property that
\begin{equation}
    \ex[g(Y) Kf(Y)] = \ex[K^*g(X) f(X)]
\end{equation}
for all $f \in \FF(\XX), g \in \FF(\YY)$. The above so-called \textit{adjoint channel} can be characterized in discrete settings as follows:
\begin{equation}
    K^*(x|y) = \dfrac{K(y|x) \mu(x)}{\mu K(y)}.
\end{equation}
Moreover, for $f = \df \nu / \df \mu$, it holds by \cite[Section 1.1 and Lemma A.1]{raginsky2016strong} that
\begin{equation}
    K^* f = \dfrac{\df (\nu K)}{\df (\mu K)}.
    \label{eq:adjchannel}
\end{equation}

In the following discussion of this paper, we will assume $1 < \alpha < \infty$ and $\mu \in \PP_*(\XX)$, $\mu K \in \PP_*(\YY)$ for $D_{\alpha}$ unless specified differently. Then it is a direct consequence that $\nu \ll \mu$ and $\nu K \ll \mu K$ for every $\nu \in \mathcal{P}(\XX)$.

\section{Lower bounds for the SDPI constants}
In this section, we will present two different versions of the lower bound for $\eta_{\alpha}(\mu, K)$. The first one, Theorem~\ref{thm:1} is universal and relates $\eta_\alpha$ to $\eta_{\chi^2}$ (similarly to the universal lower bound for $\Phi$-divergences). The second one, Theorem~\ref{thm:main}, provides a lower bound that only holds for finite alphabets but is easier to calculate. We will then show with two examples that Theorems \ref{thm:1} and \ref{thm:main} coincide in some cases. Moreover, Theorem \ref{thm:main} can sometimes help us identify whether $\eta_{\alpha}(\mu, K) = 1.$
\begin{thm}[Lower bound, Version 1]
    Given a probability distribution $\mu$ and a Markov kernel $K$, the SDPI constant for R\'enyi divergence satisfies the following inequality:
    \begin{equation}
        \eta_{\alpha}(\mu, K) \geq \eta_{\chi^2}(\mu, K).
        \label{eq:thm1res}
    \end{equation}
    \label{thm:1}
\end{thm}
\begin{proof}
    Let $\nu_{\varepsilon} = \mu + \varepsilon(\nu - \mu)$ for $0 < \varepsilon < 1$ and $f = \df \nu / \df \mu - 1$. We could derive that
    \begin{align}
        \left. \dfrac{\partial D_{\alpha}(\nu_{\varepsilon}||\mu)}{\partial \varepsilon} \right|_{\varepsilon = 0} & = \left. \dfrac{\alpha \ex_{\mu}[(1 + \varepsilon f)^{\alpha - 1} f]}{(\alpha - 1) \ex_{\mu}[(1 + \varepsilon f)^{\alpha}]}\right|_{\varepsilon = 0} = 
 0,\label{eq:1stder}\\
        \left. \dfrac{\partial^2 D_{\alpha}(\nu_{\varepsilon}||\mu)}{\partial \varepsilon^2} \right|_{\varepsilon = 0} & = \lim\limits_{\varepsilon \downarrow 0} \dfrac{1}{\varepsilon} \cdot \dfrac{\partial D_{\alpha}(\nu_{\varepsilon}||\mu)}{\partial \varepsilon} \\ &= \left. \dfrac{\alpha \ex_{\mu}[(1 + \varepsilon f)^{\alpha - 2} f^2]}{\ex_{\mu}[(1 + \varepsilon f)^{\alpha}]}\right|_{\varepsilon = 0} = \alpha \ex_{\mu}[f^{2}], \label{eq:2ndder}
    \end{align}
    where Eq. \eqref{eq:1stder} is by direct calculation of partial derivative, and Eq. \eqref{eq:2ndder} is due to the L'H\^opital's rule. Similarly,
    \begin{align}
        \left. \dfrac{\partial D_{\alpha}(\nu_{\varepsilon}K||\mu K)}{\partial \varepsilon} \right|_{\varepsilon = 0} & = 0, \\
        \left. \dfrac{\partial^2 D_{\alpha}(\nu_{\varepsilon} K||\mu K)}{\partial \varepsilon^2} \right|_{\varepsilon = 0} & = \alpha \ex_{\mu K}[(K^*f)^{2}].
    \end{align}
    Here $K^*$ is the adjoint channel of $K$. Therefore, we could apply the Taylor expansion to $D_{\alpha}(\nu_{\varepsilon}||\mu)$, which yields
    \begin{equation}
        D_{\alpha}(\nu_{\varepsilon}||\mu) = \dfrac{\alpha}{2} \ex_{\mu}[f^{2}] \varepsilon^2 + o(\varepsilon^2).
    \end{equation}
    Likewise,
    \begin{equation}
        D_{\alpha}(\nu_{\varepsilon} K||\mu K) = \dfrac{\alpha}{2} \ex_{\mu K}[(K^*f)^{2}] \varepsilon^2 + o(\varepsilon^2).
    \end{equation}
    By the definition of $\eta_{\alpha}(\mu, K)$, since $\{\nu_{\varepsilon}: \nu \neq \mu, 0 < \varepsilon < 1\} \subseteq \{\nu: \nu \neq \mu\}$, we have
    \begin{equation}
        \eta_{\alpha}(\mu, K) \geq \sup_{\nu: \nu \neq \mu} \dfrac{D_{\alpha}(\nu_{\varepsilon} K \| \mu K)}{D_{\alpha}(\nu_{\varepsilon} \| \mu)}.
    \end{equation}
    Therefore, we conclude that
    \begin{align}
        \eta_{\alpha}(\mu, K) & \geq \sup_{\nu: \nu \neq \mu} \dfrac{\ex_{\mu K}[(K^*f)^{2}] \varepsilon^2 + o(\varepsilon^2)}{\ex_{\mu}[f^{2}] \varepsilon^2 + o(\varepsilon^2)} \\
        & = \sup_{\nu: \nu \neq \mu} \dfrac{\ex_{\mu K}[(K^*f)^{2}] + o(1)}{\ex_{\mu}[f^{2}] + o(1)} = \eta_{\chi^2}(\mu, K), \label{eq:thm1conclusion}
    \end{align}
    where Eq. \eqref{eq:thm1conclusion} is derived by \cite[Remark 3.5]{raginsky2016strong}. This completes the proof.
\end{proof}

It is then immediate to see that $\eta_{\chi^2}$ is a universal lower bound for $\eta_{\alpha}$.
\begin{cor}
    Given a Markov kernel $K$, it holds
    \begin{equation}
        \eta_{\alpha}(K) \geq \eta_{\chi^2}(K).
    \end{equation}
\end{cor}
\begin{proof}
    Take the supremum over $\mu \in \mathcal{P}_{*}(\XX)$ on both sides of Eq. \eqref{eq:thm1res}.
\end{proof}

Now, we shall introduce another version of the lower bound for $\eta_{\alpha}$ which, differently from the previous one, can be written in closed form and is thus easier to compute.
\begin{thm}[Lower bound, Version 2]
Let $K$ be the kernel induced by an $n \times m$ stochastic matrix and $\mu$ be a discrete distribution on a finite alphabet $\mathcal{X}$ with $|\mathcal{X}| = n$. The SDPI constant for R\'enyi divergence satisfies the following inequality:
\begin{equation}
    \eta_{\alpha}(\mu, K) \geq \max_{1 \leq i \neq \ell \leq n} \dfrac{\mu_i \mu_{\ell}}{\mu_i + \mu_{\ell}} \sum_{j = 1}^m \dfrac{{(K_{ij} - K_{\ell j})^2}}{(\mu K)_j}.
\end{equation}
\label{thm:main}
\end{thm}
\begin{proof}
For a fixed $\nu$, set $f = \df \nu / \df \mu - 1$. We have $\ex_{\mu}[f] = 0$ and we can write
    ${D_{\alpha}(\nu K \| \mu K)}/{D_{\alpha}(\nu \| \mu)} = {\log \ex_{\mu K}[(1+K^*f)^{\alpha}]}/{\log \ex_{\mu}[(1+f)^{\alpha}]}.$
To continue, let us express $f$ as a vector and
denote its components by $f = [f_1, \ldots, f_n]^{\top}$. 
Since $\ex_{\mu}[f] = 0,$  we can express
\begin{equation}
    f_n = - \dfrac{1}{\mu_n}\sum_{r=1}^{n-1} \mu_r f_r.
    \label{eq:elln}
\end{equation}
Define the following quantities:
\begin{align}
    P(f) &:= \ex_{\mu K}[(1+K^*f)^{\alpha}] \\
    &= \sum_{j=1}^m (\mu K)_j \left(1 + \sum_{\ell=1}^{n-1} \left(K^*_{j\ell} - \dfrac{\mu_{\ell}}{\mu_n} K^*_{jn}\right)f_{\ell}\right)^{\alpha},\\
    Q(f) &:= \ex_{\mu}[(1+f)^{\alpha}]\\
    &= \sum_{j=1}^{n-1} \mu_j (1+f_j)^{\alpha} + \mu_n\left(1 - \dfrac{1}{\mu_n}\sum_{r=1}^{n-1} \mu_r f_r \right)^{\alpha}.
\end{align}
We evaluate the following derivatives, in which we denote with $\partial_i$ being the partial derivative of the corresponding function with respect to $f_i$:
\begin{align}
    \partial_i P(0) &= \alpha \sum_{j = 1}^m (\mu K)_j \left(K^*_{ji} - \dfrac{\mu_i}{\mu_n} K^*_{jn}\right) \\
    &= \alpha \sum_{j = 1}^m (\mu K)_j \left( \dfrac{\mu_i}{(\mu K)_j} K_{ij} - \dfrac{\mu_{i}}{\mu_n} \dfrac{\mu_{n}}{(\mu K)_j}K_{nj}\right) = 0, \\
    \partial_i^2 P(0) &= \alpha (\alpha - 1) \sum_{j = 1}^m (\mu K)_j \left(K^*_{ji} - \dfrac{\mu_i}{\mu_n} K^*_{jn}\right)^2 \\
    &= \alpha (\alpha - 1) \sum_{j = 1}^m \dfrac{\mu_i^2}{(\mu K)_j} \left( K_{ij} - K_{nj}\right)^2, \\
    \partial_i Q(0) &= \alpha \mu_i - \alpha \left( \dfrac{\mu_i}{\mu_n}  \right) \mu_n = 0, \\
    \partial_i^2 Q(0) &= 
    \alpha (\alpha - 1) \left( \mu_i + \dfrac{\mu_i^2}{\mu_n} \right).
\end{align}
Denote $\eta := \eta_{\alpha}(\mu, K)$. Now, for some fixed $i$, if we restrict the domain of $f$ to
\begin{equation}
    \text{Dom}_i(\varepsilon) := \{f: 0 \leq f_i \leq \varepsilon, f_s = 0 \text{ for } s \neq i\},
\end{equation}
where $\varepsilon > 0$ is small enough, then it is necessary that by the definition of $\eta$,
\begin{equation}
    \eta \geq \dfrac{\log \ex_{\mu K}[(1+K^*f)^{\alpha}]}{\log \ex_{\mu}[(1+f)^{\alpha}]}, \, f \in \text{Dom}_i(\varepsilon).
\end{equation}
Rearranging, we obtain
\begin{equation}
    \log P(f) - \eta \log Q(f) \leq 0.
    \label{eq:necessity}
\end{equation}
It is obvious that $P(0) = Q(0) = 1$. Therefore, to hold Eq. \eqref{eq:necessity}, it is necessary that
\begin{equation}
    S_1(f) := \dfrac{\partial}{\partial f_i}\left( \dfrac{P(f)}{Q^{\eta}(f)} \right) \leq 0, \, f \in \text{Dom}_i(\varepsilon).
\end{equation}
Evaluating $S_1(f)$ at $f = 0$, by the definition of partial derivative, it yields
\begin{align}
    S_1(0) &= \lim\limits_{f_i \downarrow 0} \left. \dfrac{1}{f_i} \cdot \left(\dfrac{P(f)}{Q^{\eta}(f)} - 1\right) \right|_{f = 0} \\
    &= \lim\limits_{f_i \downarrow 0} \left. \dfrac{\partial_i P(f) - \partial_i Q^{\eta}(f)}{f_i \cdot \partial_i Q^{\eta}(f) + Q^{\eta}(f)} \right|_{f = 0} \label{eq:s1second} \\ &= \partial_i P(0) - \partial_i Q(0) = 0,
\end{align}
where Eq. \eqref{eq:s1second} is due to the L'H\^opital's rule. In order to obtain $S_1(f) \leq 0$ in $\text{Dom}_i(\varepsilon)$, it is necessary to have the second derivative being negative, which is
\begin{equation}
    S_2(f) := 2\dfrac{\partial^2}{\partial f_i^2}\left( \dfrac{P(f)}{Q^{\eta}(f)} \right) \leq 0, \, f \in \text{Dom}_i(\varepsilon).
\end{equation}
Evaluating $S_2(f)$ at $f = 0$, by the L'H\^opital's rule, it yields
\begin{equation}
    S_2(0) = \left.2\lim\limits_{f_i \downarrow 0} \dfrac{P(f) - Q^{\eta}(f)}{f_i^2 \cdot Q^{\eta}(f)}\right|_{f = 0} = {\partial_i^2 P(0) - \partial_i^2 Q^{\eta}(0)},
\end{equation}
and by direct calculation,
\begin{equation}
    \partial_i^2 P(0) - \partial_i^2 Q^{\eta}(0) = \partial_i^2 P(0) - \eta \partial_i^2 Q(0),
\end{equation}
where we use the fact that $\partial_i Q(0) = 0$ and $Q^{\eta}(0) = Q^{\eta - 1}(0) = 1$. Therefore, $S_2(0) \leq 0$ is equivalent to
\begin{align}
    \eta \geq \dfrac{\partial_i^2 P(0)}{\partial_i^2 Q(0)} 
    = 
    \dfrac{\mu_i \mu_{n}}{\mu_i + \mu_{n}} \sum_{j = 1}^m \dfrac{{(K_{ij} - K_{n j})^2}}{(\mu K)_j}.
\end{align}
Since $i$ is arbitrary and, indeed, one can choose any $f_{\ell}$ to be $f_n$ in Eq. \eqref{eq:elln}, which means that $\eta$ should satisfy
\begin{equation}
    \eta \geq \max_{1 \leq i \neq \ell \leq n} \dfrac{\mu_i \mu_{\ell}}{\mu_i + \mu_{\ell}} \sum_{j = 1}^m \dfrac{{(K_{ij} - K_{\ell j})^2}}{(\mu K)_j},
\end{equation}
which is the desired result.
\end{proof}

Note that the theorem gives a lower bound that is easier to calculate compared to Theorem \ref{thm:1} and does not depend on $\alpha$. Therefore, one could take $\alpha \downarrow 1$ and obtain a lower bound of $\eta_{\text{KL}}(\mu, K)$.

We shall see in the following example that in the two-dimensional BSC, the lower bounds given in Theorem \ref{thm:1} and \ref{thm:main} coincide. As mentioned above, one could take $\alpha \downarrow 1$ to retrieve a classical lower bound for $\eta_{\text{KL}}$.
\begin{exmp}
Let $\mu = \text{Ber}(1/2)$ and $K = \text{BSC}(\varepsilon)$ for $0 \leq \varepsilon \leq 1/2$. We have $\mu K = \text{Ber}(1/2)$ in this case. Hence, Theorem \ref{thm:main} gives
\begin{align}
    \eta_{\alpha}(\mu, K) &\geq \dfrac{1}{2} \left[(K_{11} - K_{21})^2 + (K_{12} - K_{22})^2\right] \\ &= (1 - 2\varepsilon)^2 = \eta_{\chi^2}(\mu, K).\label{eq:ex1}
\end{align}
Eq. \eqref{eq:ex1} is illustrated in \cite[Example 3.1]{raginsky2016strong}. Taking $\alpha \downarrow 1$, we retrieve the well-known result described in \cite[Theorem 3.3]{raginsky2016strong} that $\eta_{\text{KL}}(\mu, K) \geq \eta_{\chi^2}(\mu, K)$.
\end{exmp}

\begin{exmp}
    Theorem \ref{thm:main} could sometimes help us to identify whether $\eta_{\alpha}(\mu, K) = 1$. Consider the following case: If there exist $i \neq \ell$, $a \neq b$ such that $K_{ia} = K_{\ell b} = 1$ and $K_{ja} = K_{j b} = 0$ for all $j \neq i, \ell$, then we obtain
    \begin{equation}
        \eta_{\alpha}(\mu, K) \geq \dfrac{\mu_i \mu_{\ell}}{\mu_i + \mu_{\ell}} \left( \dfrac{1}{\mu_i} + \dfrac{1}{\mu_{\ell}}\right) = 1.
        \label{eq:exmp2}
    \end{equation}
    Thus we conclude $\eta_{\alpha}(\mu, K) = 1$. Moreover, since the right hand side of Eq. \eqref{eq:exmp2} does not depend on $\mu$, we could further obtain $\eta_{\alpha}(K) = 1$.
\end{exmp}

\section{Properties of $\eta_2(\mu, K)$}

In this section, we investigate a special case of the SDPI constant when $\alpha = 2$. We have the following theorem.

\begin{thm}
    Let $\mu$ be a probability distribution on a finite alphabet and $K$ be a stochastic matrix with finite dimension. Whenever $\eta_{\chi^2}(\mu, K) < 1$, the SDPI constant $\eta_{2}(\mu, K)$ satisfies
    \begin{equation}
        \eta_{2}(\mu, K) = \sup \left\{\dfrac{D_{2}(\nu K \| \mu K)}{D_{2}(\nu \| \mu)}: \nu \neq \mu \text{ and } \prod_j \nu_j = 0\right\},
    \end{equation}
    which indicates that the supremum is achieved when at least one entry of $\nu$ is zero.
    \label{thm:3}
\end{thm}

\begin{proof}
    Denote $\eta := \eta_{2}(\mu, K)$ for simplicity. We formulate $\eta$ into an optimization problem as follows:
    \begin{align}
        \max_{\nu_i} \qquad &  \dfrac{\log \sum_{j} (\mu K)_j^{-1} (\nu K)_j^{2}}{\log \sum_{j} \mu_j^{-1} \nu_j^{2}} 
        \,\,\, \text{s.t.} \,\,\,   \sum_j \nu_j = 1.
    \end{align}
    Here we denote $(\mu K)_j$ as the $j$-th entry of the vector (distribution) $\mu K$. Define the Lagrangian function
    \begin{equation}
        G(\nu, \lambda) := \dfrac{\log \sum_{j} (\mu K)_j^{-1} (\nu K)_j^{2}}{\log \sum_{j} \mu_j^{-1} \nu_j^{2}} + \lambda \left( \sum_j \nu_j - 1 \right).
    \end{equation}
    The necessary condition that an extremum should satisfy is
    \begin{equation}
        \begin{aligned}
            0 = \dfrac{\partial G}{\partial \nu_i} & = \dfrac{2}{\log \sum_{j} \mu_j^{-1} \nu_j^{2} } \left(\dfrac{\sum_{j} (\mu K)_j^{-1} (\nu K)_j K_{ij}}{\sum_{j} (\mu K)_j^{-1} (\nu K)_j^{2}}  \right) \\ &- \dfrac{2\log \sum_{j} (\mu K)_j^{-1} (\nu K)_j^{2}}{\left( \log \sum_{j} \mu_j^{-1} \nu_j^{2} \right)^2}\left( \dfrac{\mu_i^{-1} \nu_i}{\sum_{j} \mu_j^{-1} \nu_j^{2}} \right) + \lambda, \, \forall i.
        \end{aligned}
        \label{eq:necessityextreme}
    \end{equation}
    A necessary condition for Eq. \eqref{eq:necessityextreme} to hold is 
        $\sum_i \nu_i {\partial G}/{\partial \nu_i} = 0.$
    Hence, we find
    \begin{equation}
        - \lambda = \dfrac{2 \log \sum_{j} \mu_j^{-1} \nu_j^{2} - 2 \log \sum_{j} (\mu K)_j^{-1} (\nu K)_j^{2}}{\left( \log \sum_{j} \mu_j^{-1} \nu_j^{2} \right)^2}.
        \label{eq:lambdalag}
    \end{equation}
    Plugging Eq. \eqref{eq:lambdalag} into Eq. \eqref{eq:necessityextreme}, we obtain
    \begin{equation}
        \begin{aligned}
            \dfrac{\partial G}{\partial \nu_i} & = \dfrac{2}{\log \sum_{j} \mu_j^{-1} \nu_j^{2} } \left(\dfrac{\sum_{j} (\mu K)_j^{-1} (\nu K)_j K_{ij}}{\sum_{j} (\mu K)_j^{-1} (\nu K)_j^{2}} - 1 \right) \\ &- \dfrac{2\log \sum_{j} (\mu K)_j^{-1} (\nu K)_j^{2}}{\left( \log \sum_{j} \mu_j^{-1} \nu_j^{2} \right)^2}\left( \dfrac{\mu_i^{-1} \nu_i}{\sum_{j} \mu_j^{-1} \nu_j^{2}} - 1\right), \, \forall i.
        \end{aligned}
    \end{equation}
    Similarly, by Eq. \eqref{eq:necessityextreme}, we must have
        $\sum_i \mu_i {\partial G}/{\partial \nu_i} = 0.$
    We obtain
    \begin{equation}
        \begin{aligned}
            &\left(1 - \dfrac{1}{\sum_{j} (\mu K)_j^{-1} (\nu K)_j^{2}}\right) \log \sum_{j} \mu_j^{-1} \nu_j^{2} \\ -& \left(1 - \dfrac{1}{\sum_{j} \mu_j^{-1} \nu_j^{2}}\right) \log \sum_{j} (\mu K)_j^{-1} (\nu K)_j^{2} = 0.
        \end{aligned}
        \label{eq:contractconstraint}
    \end{equation}
    We could rewrite Eq. \eqref{eq:contractconstraint} as
    \begin{equation}
        \begin{aligned}
            0 & = \left(\dfrac{1}{D_{2}(\nu K \| \mu K)} - \dfrac{1}{ D_{2}(\nu K \| \mu K)H_2(\nu K \| \mu K)}\right) \\ & - \left(\dfrac{1}{ D_{2}(\nu \| \mu)} - \dfrac{1}{ D_{2}(\nu \| \mu)H_{2} (\nu \| \mu)}\right).
        \end{aligned}
        \label{eq:eqtoanalyze}
    \end{equation}
    Here we define for two distributions $\nu, \mu$,
    \begin{equation}
        H_{2}(\nu \| \mu) := \sum_{j} \mu_j^{-1} \nu_j^{2}
    \end{equation}
    and we observe that $H_{2}(\nu \| \mu) = \exp{\left[D_{2}(\nu \| \mu)\right]}.$
    Define the following function for $t > 0$,
    \begin{equation}
        \varphi(t) := \dfrac{1}{t} - \dfrac{1}{t \exp (t)}.
    \end{equation}
    We would like to show that $\varphi(t)$ in non-increasing. Indeed, observing that
    \begin{equation}
        \dfrac{\df \varphi(t)}{\df t} = \dfrac{-\exp(t) + t + 1}{t^2 \exp(t)} < 0.
    \end{equation}
    Moreover, Eq. \eqref{eq:eqtoanalyze} is equivalent to
    \begin{equation}
        \varphi(D_2(\nu K \| \mu K)) - \varphi(D_2(\nu \| \mu)) = 0.
        \label{eq:monotone}
    \end{equation}
    By the non-increasing property of $\varphi$, Eq. \eqref{eq:monotone} is satisfied if and only if $D_2(\nu K \| \mu K) = D_2(\nu \| \mu)$ since we have $D_2(\nu K \| \mu K) \leq D_2(\nu \| \mu)$ always. However, we know by the definition of $\eta_{\chi^2}(\mu, K)$ that
    \begin{equation}
        \chi^2(\nu K \| \mu K) \leq \eta_{\chi^2}(\mu, K)\chi^2(\nu \| \mu).
    \end{equation}
    Therefore,
    \begin{equation}
        H_{2}(\nu K \| \mu K) - 1 \leq \eta_{\chi^2}(\mu, K)(H_{2}(\nu \| \mu) - 1).
    \end{equation}
    Since $\eta_{\chi^2}(\mu, K) < 1$ by assumption, the equality is satisfied if and only if $H_{2}(\nu \| \mu) = 1$, so that $D_2(\nu \| \mu) = 0$. In conclusion, Eq. \eqref{eq:monotone} is fulfilled if and only if
        $D_2(\nu K \| \mu K) = D_2(\nu \| \mu) = 0,$
    which is equivalent to $\nu = \mu$. However, we know by the proof Theorem \ref{thm:1} that for $\nu_{\varepsilon} = \mu + \varepsilon (\nu - \mu)$,
    \begin{equation}
        \limsup\limits_{\varepsilon \downarrow 0} \dfrac{D_{2}(\nu_{\varepsilon} K \| \mu K)}{D_{2}(\nu_{\varepsilon} \| \mu)} \leq \eta_{\chi^2}(\mu, K),
        \label{eq:subtlepoint}
    \end{equation}
    which implies that $\nu = \mu$ is a local minimum\footnote{Here is a subtle problem that the local minimum is not well defined in this case. One can argue as follows. Let $\Omega := \{\nu: \nu \neq \mu\}$ and $\Omega_{\delta} := \{\nu_{\varepsilon}: 0 < \varepsilon < \delta\}$. Consider the feasible region $\Omega \backslash \Omega_{\delta}$. The supremum is then taken on $\partial \Omega$ or $\partial \Omega_{\delta}$. However, as $\delta \downarrow 0$ we know by Eq. \eqref{eq:subtlepoint} that the supremum cannot be taken on $\partial \Omega_{\delta}$.}. Therefore, the supremum is taken at the boundary of the feasible region $\Omega = \{\nu: \nu \neq \mu\}$ except $\nu = \mu$, which is the desired result.
\end{proof}


\begin{cor}
Let
    \begin{equation}
      \tilde{\eta}_2(\mu, K) :=   \sup \left\{\dfrac{D_{2}(\nu K \| \mu K)}{D_{2}(\nu \| \mu)}: \nu \neq \mu \text{ and } \prod_j \nu_j = 0\right\}.
    \end{equation}
    Then $\eta_{2}(\mu, K)= \max \{ \tilde{\eta}_2(\mu, K), \eta_{\chi^2}(\mu, K)\}$.
\end{cor}

\begin{proof}
    This is a direct consequence of Theorems \ref{thm:1} and \ref{thm:3}.
\end{proof}

The above theorem helps us to formulate $\eta_2(\mu, K)$ explicitly when $\mu$ is a two-dimensional distribution. 
\begin{cor}
        Let $\mu = \text{Ber}(p)$ and consider the Markov kernel induced by the following $2\times 2$ matrix
    \begin{equation}
        K = \begin{bmatrix}
            1 - \varepsilon & \varepsilon \\
            \theta & 1 - \theta
        \end{bmatrix}.
    \end{equation}    
    One has that \begin{equation}
        \begin{split}
            \eta_2(\mu, K) = \max & \left\{\log_{\frac{1}{p}} \left(\dfrac{(1-\varepsilon)^2}{\mu K(\{0\})} + \dfrac{\varepsilon^2}{\mu K(\{1\})}\right),\right. \\
            & \left.\log_{\frac{1}{1-p}} \left(\dfrac{\theta^2}{\mu K(\{0\})} + \dfrac{(1-\theta)^2}{\mu K(\{1\})}\right)\right\},\label{eq:etaBinary}
        \end{split}
    \end{equation}        
    where $\mu K = \text{Ber}(p + \theta - \varepsilon p - p \theta)$.
    Moreover, if $\theta = \varepsilon$, then $K=\text{BSC}(\varepsilon)$ and
    \begin{equation}
        \eta_2(\mu, \text{BSC}(\varepsilon)) =  \log_{\frac{1}{1-p}} \left(\dfrac{\varepsilon^2}{\mu K(\{0\})} + \dfrac{(1-\varepsilon)^2}{\mu K(\{1 \})}\right).\label{eq:resforbin1}
    \end{equation}
    Furthermore,
    \begin{equation}
             \eta_2(\text{BSC}(\varepsilon)) = \log_2 \left(2 (1 - 2\varepsilon (1-\varepsilon))\right).
        \label{eq:resforbin2}
    \end{equation}
        \end{cor}
\begin{proof}
      Since $\eta_{\chi^2}(\mu, K) < 1$, Theorem \ref{thm:3} implies that the supremum is achieved by $\nu = \delta_0$ or $\nu = \delta_1$, which is exactly Eq. \eqref{eq:etaBinary}.
      
      Consider now $\theta = \varepsilon$, we could assume without loss of generality that, $0 < p, \varepsilon \leq 1/2$. Therefore, we have $p + \varepsilon - 2 \varepsilon p \leq 1/2$. By direct calculation
\begin{equation}
    \begin{aligned}
        & \left(\dfrac{\varepsilon^2}{p+\varepsilon-2\varepsilon p} + \dfrac{(1-\varepsilon)^2}{1-\varepsilon-p+2\varepsilon p}\right) \\ - & \left(\dfrac{(1-\varepsilon)^2}{p+\varepsilon-2\varepsilon p} + \dfrac{\varepsilon^2}{1-\varepsilon-p+2\varepsilon p}\right) \geq 0,
    \end{aligned}
\end{equation}
together with $\dfrac{1}{1-p} \leq \dfrac{1}{p}$ we conclude that the maximum is always achieved by the right hand side of Eq. \eqref{eq:resforbin1}.

Furthermore, by noting that $\eta_2(\mu, \text{BSC}(\varepsilon))$ is increasing with respect to $p$, we conclude that the maximum is always achieved when $p = 1/2$. Since
\begin{equation}
     \eta_2(\text{BSC}(\varepsilon)) = \sup_{0<p\leq 1/2} \eta_2(\text{Ber}(p), \text{BSC}(\varepsilon)),
\end{equation}
we arrive at Eq. \eqref{eq:resforbin2} by plugging $p = 1/2$ in Eq. \eqref{eq:resforbin1}.
\end{proof}

Figure \ref{fig:comparison} shows the plots for $\eta_{\chi^2}(\text{BSC}(\varepsilon))$, $\eta_{2}(\text{BSC}(\varepsilon))$ and $\eta_{\text{TV}}(\text{BSC}(\varepsilon))$ when $\varepsilon$ is ranging from $0$ to $1$. It is worth mentioning that we could see the lower bound $\eta_2 \geq \eta_{\chi^2}$ holds from the figure. We would also emphasize that Eqs.~\eqref{eq:resforbin1} and~\eqref{eq:resforbin2} provide closed-form formulas for $\eta_2$, with a promise for generalization to $\eta_{\alpha}$ and to arbitrary channels in the future.
\begin{figure}[h]
    \centering
    \includegraphics[width=0.43\textwidth]{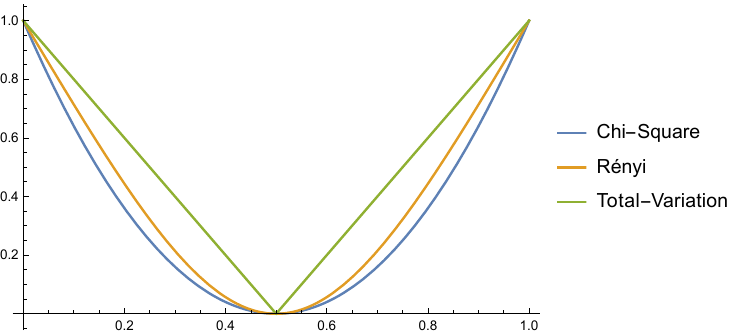}
    \caption{Plots for $\eta_{\chi^2}(\text{BSC}(\varepsilon))$, $\eta_{2}(\text{BSC}(\varepsilon))$ and $\eta_{\text{TV}}(\text{BSC}(\varepsilon))$.}
    \label{fig:comparison}
\end{figure}

\section*{Acknowledgements}

The work in this paper was supported in part by the Swiss National Science Foundation under Grant 200364.

\newpage


\bibliographystyle{IEEEtran}
\bibliography{bib}

\begin{thebibliography}{10}
\providecommand{\url}[1]{#1}
\csname url@samestyle\endcsname
\providecommand{\newblock}{\relax}
\providecommand{\bibinfo}[2]{#2}
\providecommand{\BIBentrySTDinterwordspacing}{\spaceskip=0pt\relax}
\providecommand{\BIBentryALTinterwordstretchfactor}{4}
\providecommand{\BIBentryALTinterwordspacing}{\spaceskip=\fontdimen2\font plus
\BIBentryALTinterwordstretchfactor\fontdimen3\font minus \fontdimen4\font\relax}
\providecommand{\BIBforeignlanguage}[2]{{%
\expandafter\ifx\csname l@#1\endcsname\relax
\typeout{** WARNING: IEEEtran.bst: No hyphenation pattern has been}%
\typeout{** loaded for the language `#1'. Using the pattern for}%
\typeout{** the default language instead.}%
\else
\language=\csname l@#1\endcsname
\fi
#2}}
\providecommand{\BIBdecl}{\relax}
\BIBdecl

\bibitem{ahlswede1976spreading}
R.~Ahlswede and P.~G{\'a}cs, ``Spreading of sets in product spaces and hypercontraction of the markov operator,'' \emph{The annals of probability}, pp. 925--939, 1976.

\bibitem{raginsky2016strong}
M.~Raginsky, ``Strong data processing inequalities and $\phi$-sobolev inequalities for discrete channels,'' \emph{IEEE Transactions on Information Theory}, vol.~62, no.~6, pp. 3355--3389, 2016.

\bibitem{polyanskiy2017strong}
Y.~Polyanskiy and Y.~Wu, ``Strong data-processing inequalities for channels and bayesian networks,'' in \emph{Convexity and Concentration}.\hskip 1em plus 0.5em minus 0.4em\relax Springer, 2017, pp. 211--249.

\bibitem{csizsar1967}
\BIBentryALTinterwordspacing
I.~Csisz{\'a}r, ``Information-type measures of difference of probability distributions and indirect observation,'' \emph{Studia Scientiarum Mathematicarum Hungarica}, vol.~2, pp. 229--318, 1967. [Online]. Available: \url{https://ci.nii.ac.jp/naid/10028997448/en/}
\BIBentrySTDinterwordspacing

\bibitem{liese06}
\BIBentryALTinterwordspacing
F.~Liese and I.~Vajda, ``On divergences and informations in statistics and information theory,'' \emph{IEEE Trans. Inf. Theor.}, vol.~52, no.~10, pp. 4394--4412, 2006. [Online]. Available: \url{http://dx.doi.org/10.1109/TIT.2006.881731}
\BIBentrySTDinterwordspacing

\bibitem{Witsenhausen1975ONSO}
\BIBentryALTinterwordspacing
H.~S. Witsenhausen, ``On sequences of pairs of dependent random variables,'' \emph{Siam Journal on Applied Mathematics}, vol.~28, pp. 100--113, 1975. [Online]. Available: \url{https://api.semanticscholar.org/CorpusID:123902515}
\BIBentrySTDinterwordspacing

\bibitem{anantharam2013maximal}
V.~Anantharam, A.~Gohari, S.~Kamath, and C.~Nair, ``On maximal correlation, hypercontractivity, and the data processing inequality studied by erkip and cover,'' \emph{arXiv preprint arXiv:1304.6133}, 2013.

\bibitem{cohen1993relative}
J.~E. Cohen, Y.~Iwasa, G.~Rautu, M.~B. Ruskai, E.~Seneta, and G.~Zbaganu, ``Relative entropy under mappings by stochastic matrices,'' \emph{Linear algebra and its applications}, vol. 179, pp. 211--235, 1993.

\bibitem{delMoral03}
P.~Del~Moral, M.~Ledoux, and L.~Miclo, ``On contraction properties of markov kernels,'' \emph{Probability Theory and Related Fields}, vol. 126, pp. 395--420, 01 2003.

\bibitem{van2014renyi}
T.~Van~Erven and P.~Harremo\"es, ``R{\'e}nyi divergence and {K}ullback-{L}eibler divergence,'' \emph{IEEE Transactions on Information Theory}, vol.~60, no.~7, pp. 3797--3820, 2014.

\bibitem{sasonV16}
I.~Sason and S.~Verdú, ``$f$ -divergence inequalities,'' \emph{IEEE Transactions on Information Theory}, vol.~62, no.~11, pp. 5973--6006, 2016.

\bibitem{esposito2023concentration}
A.~R. Esposito and M.~Mondelli, ``Concentration without independence via information measures,'' 2023.

\bibitem{gross75}
\BIBentryALTinterwordspacing
L.~Gross, ``Logarithmic sobolev inequalities,'' \emph{American Journal of Mathematics}, vol.~97, no.~4, pp. 1061--1083, 1975. [Online]. Available: \url{http://www.jstor.org/stable/2373688}
\BIBentrySTDinterwordspacing

\bibitem{esposito2023lower}
A.~R. Esposito, A.~Vandenbroucque, and M.~Gastpar, ``Lower bounds on the {B}ayesian risk via information measure,'' \emph{arXiv preprint arXiv:2303.12497}, 2023, {\it Journal of Machine Learning Research}, accepted for publication.

\bibitem{vempala2019Langevin}
\BIBentryALTinterwordspacing
S.~Vempala and A.~Wibisono, ``Rapid convergence of the unadjusted langevin algorithm: Isoperimetry suffices,'' in \emph{Advances in Neural Information Processing Systems}, H.~Wallach, H.~Larochelle, A.~Beygelzimer, F.~d\textquotesingle Alch\'{e}-Buc, E.~Fox, and R.~Garnett, Eds., vol.~32.\hskip 1em plus 0.5em minus 0.4em\relax Curran Associates, Inc., 2019. [Online]. Available: \url{https://proceedings.neurips.cc/paper_files/paper/2019/file/65a99bb7a3115fdede20da98b08a370f-Paper.pdf}
\BIBentrySTDinterwordspacing

\bibitem{chen2022Sampling}
\BIBentryALTinterwordspacing
Y.~Chen, S.~Chewi, A.~Salim, and A.~Wibisono, ``Improved analysis for a proximal algorithm for sampling,'' in \emph{Proceedings of Thirty Fifth Conference on Learning Theory}, ser. Proceedings of Machine Learning Research, P.-L. Loh and M.~Raginsky, Eds., vol. 178.\hskip 1em plus 0.5em minus 0.4em\relax PMLR, 02--05 Jul 2022, pp. 2984--3014. [Online]. Available: \url{https://proceedings.mlr.press/v178/chen22c.html}
\BIBentrySTDinterwordspacing

\end{thebibliography}


\end{document}